\documentclass[reqno]{amsart}
\usepackage{amssymb}

\usepackage{hyperref}
\usepackage{pdfsync}

\numberwithin{equation}{section}
\newtheorem{theorem}{Theorem}[section]

\newtheorem{lemma}[theorem]{Lemma}

\newtheorem{remark}[theorem]{Remark}

\DeclareMathOperator{\supp}{supp}
\DeclareMathOperator{\tr}{tr}

\newcommand\R{\mathbb R}
\newcommand\N{\mathbb N}

\newcommand\Z{\mathbb Z}

\newcommand\e{\mathrm{e}}

\renewcommand\P{\mathbb P}
\newcommand\E{\mathbb E}

\newcommand\eps{\varepsilon}

\newcommand{\pr}{\prime}

\newcommand{\bom}{{\boldsymbol{{\omega}}}}

\newcommand\beq{\begin{equation}}
\newcommand\eeq{\end{equation}}

\newcommand{\abs}[1]{\left\lvert #1 \right\rvert}
\newcommand{\norm}[1]{\left\lVert #1 \right\rVert}
\newcommand{\scal}[1]{\left\langle #1 \right\rangle}
\newcommand{\set}[1]{\left\{ #1 \right\}}
\newcommand{\pa}[1]{\left( #1 \right)}

\newcommand{\eq}[1]{\eqref{#1}}

\newcommand{\up}[1]{^{(#1)}}

\begin{document}

\title[Lifshitz tails for the DOS]{Lifshitz tails estimate for the density of states of the Anderson model}

\author[J.-M. Combes]{Jean-Michel Combes}
\address[Combes]{CPT, CNRS, Luminy Case 907, Cedex 9, F-13288 Marseille, France}
\email{combes@cpt.univ-mrs.fr}

\author[F. Germinet]{Fran\c cois Germinet}
\address[Germinet]{Universit\'e de Cergy-Pontoise,
CNRS UMR 8088, IUF, D\'epartement de Math\'ematiques,
F-95000 Cergy-Pontoise, France}
\email{germinet@math.u-cergy.fr}

\author[A. Klein]{Abel Klein}
\address[Klein]{University of California, Irvine,
Department of Mathematics,
Irvine, CA 92697-3875,  USA}
\email{aklein@uci.edu}


\thanks{JM.C. and F.G. were supported by ANR BLAN 0261. A.K was  supported in part by NSF Grant DMS-1001509.}


\begin{abstract}
We prove an upper bound for the (differentiated) density of states of the Anderson model  at the bottom of the spectrum. The density of states is shown to exhibit the same Lifshitz tails upper bound as the integrated density of states.
\end{abstract}

\maketitle



\section{Introduction and main result}

We consider the Anderson model, the simplest random Schr\"odinger operator, given by the random Hamiltonian 
\beq\label{AndH}
H_{\bom}: = -\Delta+
V_{\bom} \quad \text{on}\quad \ell^2(\mathbb{Z}^d),
\eeq
$\Delta$ is the $d$-dimensional discrete Laplacian operator  and  $V_{\bom}$ is the random potential given by $V_{\bom}(j)=\omega_{j}$ for $j \in \Z^{d}$, where  $\bom=\{ \omega_j \}_{j\in
\Z^d}$  is a family of independent 
identically distributed random
variables  whose  common probability 
distribution $\mu$ is non-degenerate, with support bounded from below, and    has a bounded density $\rho$.  (The requirement $ \inf \supp \mu > -\infty$ is equivalent to requiring $H_{\bom} $ to be bounded from below with probability one.) Note that
\beq
V_{\bom} = 
\sum_{j \in \Z^d} \omega_j \Pi_j ,     \label{AndV}
\eeq
where $ \Pi_j$ is the orthogonal projection onto $\delta_{j}$, the delta function at site $j$.

The integrated density of states (IDS) of this celebrated model is known to exhibit Lifshitz tails behavior as the bottom of its spectrum (e.g.,  \cite{CL,Ki,PF}), a property that can be interpreted as a first signature of localization.  

In the physics literature there is much interest on the density of the IDS, the density of states (DOS).  There  is an implicit assumption that the IDS $N(E)$ is absolutely continuous, and hence has a density $n(E)$ given by its almost everywhere derivative.  Very few mathematical results are available for the DOS of   random   Schr\"odinger 
operators.  For nice enough models, like the  Anderson model as above, the  existence of the DOS  is a consequence of the celebrated Wegner estimate (see \eqref{hypWegner} below), which also shows the  DOS to be bounded. We note that while the absolute continuity of the IDS has been  known for a long time for the Anderson model \cite{W,FS}, this is a more recent result for models in the continuum \cite{CH,CHK2}. In dimension $d=1$ this DOS is known to be regular \cite{ST,CK,VSW}.  For arbitrary dimension $d$,    the DOS of the Anderson model described above is known to be regular at high disorder, a result obtained using supersymmetric methods by \cite{BCKP}. Regularity of the DOS is an open question  for models in the continuum.

In this article, we show that the DOS exhibits the same Lifshitz tails  upper bound as the IDS. To our best knowledge, this is the first time that such a bound is proved for the DOS of a random Schr\"odinger operator.

There is another reason to study  the DOS $n(E)$.   As proved by Minami \cite{Mi}, in the localization region  the properly rescaled eigenvalues of a discrete  Anderson Hamiltonian are  distributed as a Poisson point process with intensity $n(E)$.  Thus our results (see \eq{nelowerbd}) estimate the intensity of these Poisson point processes.

An  Anderson Hamiltonian $H_{\bom}$  is a $\Z^d$-ergodic family of
random self-adjoint operators.
It follows from standard results (cf. \cite{KiM})
that there exists fixed subsets $\Sigma$,  $\Sigma_{\mathrm{pp}}$, $\Sigma_{\mathrm{ac}}$  and $\Sigma_{\mathrm{sc}}$ of $\R$ so that the spectrum $\sigma(H_{\bom})$
of $H_{\bom}$,  as well as its pure point, 
absolutely continuous, and singular continuous  components,
are equal to these fixed sets with probability one. This non-random spectrum is given by 
\beq
\Sigma = \sigma(-\Delta)+ \supp \mu=[0,4d] +\supp \mu.
\eeq
Note $\inf \Sigma=\inf \supp \mu > -\infty$. 

If we set   $\omega_{j}^{\pr}= \omega_{j}-\inf \supp \mu$,  the new common probability distribution $\mu^{\pr}$ satisfies   $\inf \supp \mu^{\pr}=0$, and  we have  $H_{\bom}= H_{\bom^{\pr}} + \inf \supp \mu $.  Thus, without loss of generality we   assume from now on that
\beq
\inf\Sigma = \inf \supp \mu =0.
\eeq

Lifshitz tails and localization are known to hold at the bottom of the spectrum \cite{FMSS,vDK,Ki}. If the support of $\mu$ is also bounded from above, then so is $\Sigma$, and Lifshitz tails and localization also hold at the top, i.e.,  upper edge, of the spectrum.  In this case the results of this paper  also hold at the top of the spectrum.

The integrated density of states (IDS), $N(E)$,  can be written as (e.g., \cite{Ki})
\beq\label{defN(E)}
N(E) =\E \set{\tr \pa{\Pi_0 \chi_{]-\infty,E]}(H_\omega) \Pi_0}}=\E \set{\tr \pa{\Pi_0 \chi_{[0,E]}(H_\omega) \Pi_0}}.
\eeq 

The  Anderson model  is known to satisfy the following Lifshitz tails estimate, which asserts that the IDS has an exponential fall off as one approaches the edges of $\Sigma$. At the bottom of the spectrum, i.e., at energy $E=0$,  the IDS  satisfies (e.g., \cite{Ki})
\beq\label{LT}
\lim_{E \downarrow 0}     \frac{\log \abs{\log    N(E)}}{\log E}   
\le - \frac d2 .
\eeq
Equality is actually known to hold in \eqref{LT}.  Since $N(E)$ is an increasing function, it has a derivative $n(E):= N^{\pr}(E)$  almost everywhere, which is the density of states.

We prove 
\begin{theorem}\label{thmLT}
Let $H_\omega$ be an Anderson Hamiltonian.  Then there exists a Borel set $\mathcal{N}\subset [0,1]$ of zero Lebesgue measure, such that
\beq\label{LTdos}
\lim_{E \downarrow 0; \ E \notin \mathcal{N}}     \frac{\log \abs{\log    n(E)}}{\log E}   
\le - \frac d2 .
\eeq
\end{theorem}

\begin{remark}
The same Lifshitz tails estimate holds for models in the continuum; see \cite{CGK2}.
\end{remark}

The proof of Theorem~\ref{thmLT} takes advantage of a new double averaging procedure introduced in \cite[Theorem~2.2]{CGK} to extract better control on the constant in the  Wegner estimate. Theorem~\ref{thmLT} will be an immediate consequence of Theorem~\ref{thmLTW}.


\section{Finite volume operators, the integrated density of states, and the Wegner estimate}

Finite volume operators will be defined for 
finite boxes \beq
\Lambda=\Lambda_L(j):= j +\left[-\tfrac L 2, \tfrac L 2\right[^d,
\eeq
where $j\in \Z^d$ and $L \in 2\N$,  $ L > 1$.   Given such $\Lambda$, we will consider the random Schr\"odinger operator $H_\bom^{(\Lambda)}$ on $\ell^2(\Lambda)$ given by the restriction of the Anderson Hamiltonian  $H_\bom$ to $\Lambda$ with periodic boundary condition. To do so, we identify $\Lambda$ with the   torus $ \Z^{d}\slash L\Z^{d}$
in the usual way, and  define finite volume operators  
\begin{align}\label{finvolH}
H_{\bom}\up{\Lambda} := - \Delta\up{\Lambda}+ V_{\bom}\up{\Lambda} \quad \text{on}   \quad \ell^{2}(\Lambda),
\end{align}
where 
$\Delta\up{\Lambda}$ is the  Laplacian on $\Lambda$ with periodic boundary condition,
and the random potential $V_{\bom}\up{\Lambda}$ is the restriction of $V_{\bom\up{\Lambda}}$ to $\Lambda$, where, given $\bom=\set{\omega_i}_{i \in \Z^d}$ , we define $\bom\up{\Lambda}=\set{\omega\up{\Lambda}_i}_{i \in \Z^d}$ by
\beq\begin{split}
\omega\up{\Lambda}_i &=\omega_i \quad \text{if}  \quad i \in \Lambda, \\
\omega\up{\Lambda}_i &=\omega\up{\Lambda}_k  \quad \text{if}  \quad k-i \in L\Z^d . \label{periodomeg}
\end{split}\eeq

The  finite volume random operator $ H_{\bom}\up{\Lambda}$ is covariant with respect to translations in the torus. 
If $B \subset \R$ is a Borel set, we write $P_\bom^{(\Lambda)}(B):=\chi_B\pa{H_\bom^{(\Lambda)}}$ and $P_\bom(B):=\chi_B(H_\bom)$ for the spectral projections.

The finite volume operator $H_\bom^{(\Lambda)}$ is a finite dimensional operator, and hence its ($\bom$-dependent) spectrum consists of a finite number of  isolated  eigenvalues with finite multiplicity.  These finite volume  operators satisfy  a Wegner estimate \cite{W,FS} (see also \cite{CGK0}), which provides informations on the number of eigenvalues in a given energy  interval: given $E_0 > 0$, there exists a  constant $K_W(E_0)$, independent of $\Lambda$, such that for  all intervals $I\subset [0,E_0]$ we have
\beq\label{hypWegner}
\E \set {\tr P_\bom^{(\Lambda)}(I)} \le K_W(E_0)  \|\rho\|_\infty \abs{I}\abs{\Lambda}.
\eeq
The Wegner estimate is an an immediate consequence of the following spectral averaging property (e.g., \cite{CHK2,CGK0}) :  for any $k\in\Lambda$,
\beq\label{hypSA}
\E_{\omega_k} \set {\langle \delta_k, P_\bom^{(\Lambda)}(I) \delta_k\rangle}  = \E_{\omega_k} \set {\tr \Pi_k P_\bom^{(\Lambda)}(I) \Pi_k} \le  \|\rho\|_\infty \abs{I},
\eeq
where $(\delta_k)_{k\in\Z^d}$ stands for the canonical basis (so $\Pi_k=  |\delta_k\rangle \langle \delta_k|$).
It follows that for the Anderson model  \eq{hypWegner} holds with
\beq\label{boundK}
 K_W(E_0) \le 1 \quad \text{for all} \quad  E_0>0.
\eeq
We shall prove that at the bottom of the spectrum, in the Lifshitz tails region, the constant  $K_W(E_0)$ falls off in the same way as the IDS.

We set
\beq\label{NELambda}
N_{\bom}^{(\Lambda)}(E):=  \abs{\Lambda}^{-1} \tr  \, \chi_{]- \infty,E]}\pa{H_{\bom}\up{\Lambda}},
\eeq
and  recall  that  (e.g., \cite{CL,Ki,PF}) for $\P$-a.e. $\bom$ we have, using the fact that $N(E)$ is a continuous function by \eq{hypWegner}, 
 \beq \label{N(E)}
N(E)= \lim_{L \to \infty} N_{\bom}^{(\Lambda_{L}(0))}(E)\quad \text{for all}\quad E \in \R,
\eeq
where $N(E)$ is
 the integrated density of states (IDS) given in \eq{defN(E)}).  
 Setting
\beq\label{NELaExp}
N^{(\Lambda)}(E):=\E\set{N_{\bom}^{(\Lambda)}(E) }=   \E\set{\tr\pa{\Pi_j\chi_{]-\infty,E]}(H_{\bom}\up{\Lambda}) \Pi_j}} \quad \text{for }\quad j \in \Lambda,
\eeq
the last equality holding in view of the periodic boundary condition,
 it follows that
 \beq \label{N222}
N(E)= \lim_{L \to \infty} N^{(\Lambda_{L}(0))}(E)= \E\set{\tr\pa{\Pi_0\chi_{]-\infty,E]}(H_{\bom}\up{\Lambda_{L}(0)}) \Pi_0}}
\eeq
for all $E\in \R$.

Combining  \eqref{hypWegner} and \eqref{N222}, we conclude that, if $N(E)$ is differentiable at $E$, which is true for a.e.\ $E$, we have
\beq\label{nEestW}
n(E):=N^\pr(E) \le   K_W(E)  \|\rho\|_\infty .
\eeq
Moreover, to obtain \eq{nEestW}, it suffices to have  the Wegner estimate  \eqref{hypWegner} for boxes $\Lambda_{L_{n}}(0)$ with  $L_{n} \to \infty$.    Thus Theorem~\ref{thmLT} follows immediately from the following result.

\begin{theorem}\label{thmLTW}
Let $H_\omega$ be an Anderson Hamiltonian.  Then there is an energy $E_0>0$, and for a.e. $E \in ]0,E_{0}]$  and all  $\eps\in ]0,\frac d 2[$ there exists a scale $L(E,\eps)<\infty$, such that given  $L \in 4\N$ with  $L \ge L(E,\eps)$,  the Wegner estimate  \eqref{hypWegner}  holds in all boxes  $\Lambda=\Lambda_L$  for all intervals $I \subset [0,E]$ with a constant $K_W(E)$ satisfying 
\beq
K_W(E) \le C_{d,\eps} \,\e^{-E^{-\frac d2+\eps}},
\eeq
for some constant $C_{d,\eps}<\infty$.  As a consequence, we have
\beq\label{nelowerbd}
n(E) \le C_{d,\eps}\norm{\rho}_{\infty} \e^{-E^{-\frac d2+\eps}} \quad \text{for a.e.}\quad E \in ]0,E_{0}].
\eeq
\end{theorem}

\section{The proof of Theorem~\ref{thmLTW}}

We borrow from \cite{CGK2} the following observation. 
\begin{lemma}[\cite{CGK2}]\label{lemobs}
Let $H=H_0+W$, where $H ,H_0$ are semi-bounded  self-adjoint operators, say $H,H_{0}\ge -\Theta$ for some $\Theta>0$,  such that    $\pa{H+ \Theta+1}^{-p}$
is a trace class operator for some $p >0$, and $W$ is a  bounded  self-adjoint operator. Given $E_0 \in \R$, 
let $f,g$ be  bounded Borel measurable nonnegative functions such that
$f= \chi_{(-\infty, E_0]} f$ and $\chi_{(-\infty, E_0]}\le  g\le 1$.  Then   $ f(H)W$ and $ f(H) W g(H_0)$ are  trace class operators, and 
 \begin{equation}
\tr f(H)W \le \tr f(H) W g(H_0).
\end{equation}
\end{lemma}
The proof is elementary. It consists in proving that $ \tr f(H) W (1-g(H_0))\le 0$, using $W=H-H_0$.

\begin{proof}[Proof of Theorem~\ref{thmLTW}]  All the operators  in this proofs will be finite volume operators on a box $\Lambda=\Lambda_{L}(0)$ as defined in \eqref{finvolH}-\eqref{periodomeg}.  We will require $L \in 4 \N$.

We set  $\omega_0^\perp=\bom \setminus\{\omega_{0}\}$ and     $H_{\omega_0^\perp}=H_{\bom}- \omega_{0 }\Pi_{0}$. Given $E>0$, we fix a $C^{\infty}$ real-valued non-increasing function $f_{E}$ on $\R$, such that  $f_{E}(t)=1$ for $t\le E$, $f_{E}(t)=0$ for $t\ge 2E$,   and $\abs{f^{(j)}(t)}\le C E^{-j}$ for all $t\in \R$ and $j=1,2,\ldots,2d+3$, where $C$ is a constant independent of $E$.  We let
$\tilde{P}_0=\tilde{P}\up{\Lambda}_{\omega_0^\perp,E}= f_{E}(H\up{\Lambda}_{\omega_0^\perp})$.

Given an interval $I \subset [0,E]$, it follows from Lemma \ref{lemobs} that
\begin{align}\notag
\tr {P\up{\Lambda}_{\bom}}(I) \Pi_0
& \le  \tr {P\up{\Lambda}_{\bom}}(I) \Pi_0 \tilde{P}_0
 = \sum_{k\in \Lambda} \tr {P\up{\Lambda}_{\bom}}(I) \Pi_0 \tilde{P}_0 \Pi_k\\ \label{sumtr}
& \le \sum_{k\in \Lambda} \| {P\up{\Lambda}_{\bom}}(I) \Pi_0\|_2 \| {P\up{\Lambda}_{\bom}}(I) \Pi_k\|_2 \|\Pi_0 \tilde{P}_0 \Pi_k \| \\
& \le \frac12 \sum_{k\in \Lambda} \set{\tr  (\Pi_0 {P\up{\Lambda}_{\bom}}(I) \Pi_0) + \tr (\Pi_k {P\up{\Lambda}_{\bom}}(I) \Pi_k) } \|\Pi_0 \tilde{P}_0 \Pi_k \|. \notag
\end{align}
Next, for any given $k$ (including $k=0$) there exists $k_0\in \Z^d$, so that, defining the sublattice $\Gamma_{k}=k_0+(4\Z)^d\subset\Z^d$, we have $0,k\not\in \Gamma$. 
Since we required $L \in 4\N$,  we have $ \abs{\Gamma\cap\Lambda} =  4^{-d}|\Lambda|$. We define $\bom_{\Gamma}=\set{\omega_{\gamma}}_{\gamma\in \Gamma}$. We have, using the fast decay of the kernel of smooth functions of Schr\"odinger operators (e.g. \cite{GKdecay}),  
 using the notation  $\scal{k}=\sqrt{1 + \abs{k}^{2}}$, and letting $r$ stand for either $0$ or $k$,
\begin{align}
\|\Pi_0 \tilde{P}_0 \Pi_k \| \notag
& \le \|\Pi_0 \tilde{P}_0 \Pi_k \|^{\frac12} \|\Pi_r \tilde{P}_0  \|^{\frac 12}=  \|\Pi_0 \tilde{P}_0 \Pi_k \|^{\frac12} \|\Pi_r \tilde{P}_0  \Pi_r\|^{\frac 14}\\ \label{usepositivity}
& \le  \frac{C_d}{ E^{d+\frac 3 2}\scal{k}^{d+1}}\pa{  \tr \Pi_r \tilde{P}_0 \Pi_r }^{\frac 14} \\ \notag
 &\le  \frac{C_d}{ E^{d+\frac 3 2}\scal{k}^{d+1}} \pa{  \e^{2tE}\tr  \Pi_r  \e^{- t H\up{\Lambda}_{\bom_{0}^{\perp}}} \Pi_r   }^{\frac 14} \\ \notag
& \le  \frac{C_d}{ E^{d+\frac 3 2}\scal{k}^{d+1}} \pa{  \e^{2tE} \tr  \Pi_r  \e^{- t H\up{\Lambda}_{\bom_{\Gamma_k}}} \Pi_r   }^{\frac 14} , 
\end{align}
for all $t>0$, where we used  $r\not\in \Gamma_k$ and  a positivity preserving argument like \cite[Lemma~2.2]{BGKS} to get the last inequality.
Hence, again letting $r$ stand for either $0$ or $k$, so  $r\not\in \Gamma_k$, using the spectral averaging \eqref{hypSA} with the bound \eqref{boundK} and  \eqref{usepositivity} leads to
 \begin{align}\notag
&\E\pa{ \tr \set{  \Pi_r{P\up{\Lambda}_{\bom}}(I) \Pi_r}  \|\Pi_0 \tilde{P}_0 \Pi_k \| } \\ \notag
& \qquad \le \frac{C_d }{ E^{d+\frac 3 2}\scal{k}^{d+1}}  \E_{\omega_{r}^\perp}   \set{ \pa{\e^{2tE} \tr  \Pi_r  \e^{- t H\up{\Lambda}_{\bom_{\Gamma_k}}} \Pi_r   }^{\frac14} \E_{\omega_{r}}\set{ \tr  \Pi_r {P\up{\Lambda}_{\bom}}(I) \Pi_r } }
\\ \notag 
& \qquad \le \frac{C_d}{ E^{d+\frac 3 2}\scal{k}^{d+1}}  \|\rho\|_\infty |I|  \e^{\frac t 2 E}\pa{\E_{\bom_{\Gamma_k}}\set{ \tr  \Pi_r  \e^{- t H\up{\Lambda}_{\bom_{\Gamma_k}}} \Pi_r}^\frac14  }\\
& \qquad \le \frac{C_d}{ E^{d+\frac 3 2}\scal{k}^{d+1}}  \|\rho\|_\infty |I|  \e^{\frac t 2 E}\pa{\E_{\bom_{\Gamma_k}}\set{ \tr  \Pi_r  \e^{- t H\up{\Lambda}_{\bom_{\Gamma_k}}} \Pi_r} }^\frac14 . \label{boundk}
\end{align}

Thanks to \eqref{sumtr} and \eqref{boundk}, it now suffices to bound  the quantity
\begin{align}
\E_{\bom_{\Gamma_k}}\set{ \tr  \Pi_r  \e^{- t H\up{\Lambda}_{\bom_{\Gamma_k}}} \Pi_r } \quad \text{with}\quad r=0,k.
\end{align}
To alleviate notations, we write from now on $\Gamma=\Gamma_k$.  For all  $t >0 $ we have 
\begin{align}
\notag
 \tr \Pi_{r} \e^{- tH_{\bom_{\Gamma}}}\Pi_{r}
 &
 \le   \tr \Pi_{r} \chi_{]-\infty,4E]}(H\up{\Lambda}_{\bom_{\Gamma}})\Pi_{r} + \e^{-4 tE}\tr \Pi_{r} \notag \\
 & =  \tr \Pi_{r} \chi_{]-\infty,4E]}(H\up{\Lambda}_{\bom_{\Gamma}})\Pi_{r} +  \e^{-4tE}. \label{boundsublat}
 \end{align}
By the $\Gamma$-ergodicity   of $H\up{\Lambda}_{\bom_{\Gamma}}$, and taking advantage of the periodic boundary condition, we have
 \begin{align}
 \E\set{ \tr \Pi_r \chi_{]-\infty,4E]}(H\up{\Lambda}_{\bom_{\Gamma}})\Pi_r}
 & = \frac 1 {\abs{\Gamma\cap\Lambda}} \sum_{\gamma\in(r+\Gamma)\cap\Lambda} \E_{\bom_\Gamma}\set{ \tr \Pi_{\gamma} \chi_{]-\infty,4E]}(H\up{\Lambda}_{\bom_{\Gamma}})\Pi_{\gamma}}
 \notag \\
 &\le \frac {4^{d}} {\abs{\Lambda}} \E_{\bom_\Gamma}\set{ \tr  \chi_{]-\infty,4E]}(H\up{\Lambda}_{\bom_{\Gamma}})}  .\label{transl}
  \end{align}

We now use the Lifshitz tails estimate for $H_{\bom_\Gamma}$ to bound $\E_{\bom_\Gamma}\set{ \tr  \chi_{]-\infty,4E]}(H\up{\Lambda}_{\bom_{\Gamma}})}$. Note that $\Gamma$ is a strict sublattice of $\Z^d$, so  we lack the so called covering condition. The Lifshitz tails estimate  \eq{LT} is nevertheless valid for $H_{\bom_\Gamma}$ (e.g. \cite{Ki}), and it implies that for all $\eps\in ]0, \frac d 2[$ there is an energy $E_\eps >0$  such that (on $\Z^{d}$)
\beq \label{lifest}
N_{\Gamma}(E)= \E_{\bom_\Gamma}\set{ \tr \chi_{]-\infty,E]}(H_{\bom_{\Gamma}})}  \le  \e^{- E^{- \frac d2 +\eps} }\quad \text{for all} \quad E \le E_\eps.
\eeq 

Given a box $\Lambda_{L}=\Lambda_{L}(0)$, we set, similarly to \eq{NELambda} and \eq{NELaExp}, 
\beq
N_{\bom_{\Gamma}}\up{\Lambda_{L}}(E)= \abs{\Lambda_{L}}^{-1}\ \tr  \chi_{]-\infty,E]}(H_{\bom_{\Gamma}}\up{{\Lambda_{L}}}), \quad N_{{\Gamma}}\up{\Lambda_{L}}(E)= \E_{\bom_{\Gamma}}\set{N_{\bom_{\Gamma}}\up{\Lambda_{L}}(E)}.
\eeq
Since $H_{\bom_{\Gamma}}$ is  $\Gamma$-ergodic, for $\P$-a.e. $\bom_{\Gamma}$
we have \cite{CL,PF}
\begin{align} \label{N(E)Gamma}
N_{\Gamma}(E)= \lim_{L \to \infty} N_{\bom_{\Gamma}}^{(\Lambda_{L})}(E)\quad \text{for a.e.}\quad E \in \R,
\end{align}
so we conclude
\beq\label{NGammaconv}
N_{\Gamma}(E)= \lim_{L \to \infty} N_{{\Gamma}}^{(\Lambda_{L})}(E)\quad \text{for a.e.}\quad E \in \R.
\eeq
(Compare with   \eq{N(E)} and \eq{N222}, where convergence holds for all $E$.  The difference is the lack of a Wegner estimate for  $H_{\bom_{\Gamma}}$.)
It follows that for a.e.\ $E\in ]0, E_{\eps}]$
 there exists $L(E,\eps)< \infty$ such that for all $L\ge L(E,\eps)$ we have 
\beq \label{estNvol}
N_{\Gamma}\up{\Lambda_{L}}(E)\le 2 N_{\Gamma}(E)\le 2   \e^{- E^{- \frac d2 +\eps} }.
\eeq

Thus, for a.e.\ $E\in ]0, \frac 1 4 E_{\eps}]$ and  $L\ge L(E,\eps)$, 
combining \eqref{boundsublat}, \eq{transl}, and  \eq{estNvol}, we get
\begin{align}
\E\set{\tr \Pi_{r}\tilde{P}_{0}\up{\Lambda_{L}} \Pi_{r} }
\le 2\cdot 4^{d} \e^{- (4E)^{- \frac d2 +\eps} }+ \e^{- 4tE} .
\end{align}
Choosing  $t=t_{E}$ by
\beq
 \e^{- 4tE} = 2\cdot 4^{d} \e^{- (4E)^{- \frac d2 -1 +\eps} }, \quad\text{i.e.,}\quad
 t= (4E)^{- \frac d2-1  +\eps}  -  (1 +2d)( \log 2)(4E)^{-1},
\eeq
and letting $E_{\eps}^{\pr}= \max \set{E \le \tfrac 1 4E_\eps; \ t_{E}> 0 }$, we conclude that
 for a.e.\ $E\in ]0,  E_{\eps}^{\pr}]$ and  $L\ge L(E,\eps)$ we have
\begin{align}\label{trPiLT}
\E\set{\tr \Pi_{r}\tilde{P}_{0}\up{\Lambda_{L}}  \Pi_{r} }
\le  4^{d+1} \e^{- (4E)^{- \frac d2  +\eps} }=  2\e^{- 4t_{E}E} .
\end{align}
Combining \eq{sumtr}, \eq{boundk}, and \eq{trPiLT}, we conclude that, for a.e.\ $E\in ]0,  E_{\eps}^{\pr}]$ and  $L\ge L(E,\eps)$ we have, for all intervals $I \subset [0,E]$,
\begin{align}
&\E \set{\tr {P_{\bom}}\up{\Lambda_{L}}(I) \Pi_0} \le      \sum_{k\in \Lambda}  \frac{C_d }{ E^{d+\frac 3 2}\scal{k}^{d+1}}  \|\rho\|_\infty |I|  \e^{\frac 1 2  t_{E}  E}\pa{ 2\e^{- 4t_{E}E}}^{\frac 1 4}\\
& \qquad \le C^{\pr}_{d} E^{-d-\frac 32}    \e^{-\frac 1 2 t_{E}  E}      \|\rho\|_\infty |I| 
= C^{\pr\pr}_{d} E^{-d-\frac 32}    \e^{- \frac 1 8(4E)^{- \frac d2  +\eps} }      \|\rho\|_\infty |I| .
\notag
\end{align}

Theorem~\ref{thmLTW}    follows. \end{proof}


\end{document}